\documentclass[journal]{IEEEtran}

\usepackage{amsfonts}
\usepackage{amsmath}
\usepackage{theorem}
\usepackage{array}
\usepackage[a4paper]{geometry}
\usepackage{amssymb}
\usepackage{graphics}
\usepackage[usenames]{color}
\usepackage[all]{xy}
\usepackage{graphicx}
\usepackage{boxedminipage}

\newtheorem{theorem}{Theorem}[section]

\newtheorem{corollary}[theorem]{Corollary}

\newtheorem{remark}[theorem]{Remark}

\newtheorem{construction}[theorem]{Construction}

\newtheorem{lemma}[theorem]{Lemma}

\newenvironment{proof}[1][Proof]{\textbf{#1.} }{\ \rule{0.5em}{0.5em}}

\def\PP{\mathbb{P}}

\def\F{\mathbb{F}}

\def \mF {\mathcal{F}}

\def \mL {\mathcal{L}}

\def\tP{\textsf{P}}
\def\tR{\textsf{R}}

\def\mL{{\mathcal L}}

\newcommand{\Ga}{\alpha}
\newcommand{\Gb}{\beta}

\newcommand{\Gg}{\gamma}

\newcommand{\Go}{\omega}     \newcommand{\GO}{\Omega}

\def\Supp{{\rm Supp}}

\def \ba {{\bf a}}

\def \bc {{\bf c}}

\def \bx {{\bf x}}

\def \bu {{\bf u}}
\def \bv {{\bf v}}
\def \bo {{\bf 0}}
\def\Supp{{\rm Supp}}
\def\res{{\rm res}}
\def\RS{\mathsf{GRS}}
\title{Explicit MDS Codes with Complementary Duals}

\author{Peter Beelen and Lingfei Jin}

\date{}

\onecolumn

\begin{document}
\maketitle

\begin{abstract}
In 1964, Massey introduced a class of codes with complementary duals which are called Linear Complimentary Dual (LCD for short) codes. He showed that LCD codes have applications in communication system, side-channel attack (SCA) and so on. LCD codes have been extensively studied in literature. On the other hand, MDS codes form an optimal family of classical codes which have wide applications in both theory and practice. The main purpose of this paper is to give an explicit construction of several classes of LCD MDS codes, using tools from algebraic function fields. We exemplify this construction and obtain several classes of explicit LCD MDS codes for the odd characteristic case.
\end{abstract}


\section{Introduction}

Linear Complementary Dual (LCD) cyclic codes over finite fields were first introduced  and studied  by  Massey in 1964 and they were called reversible codes in  \cite{Mass64}. LCD codes have been used in many applications. Besides the applications in communication systems, data storage, Carlet and Guilley  \cite{CG14} showed that LCD codes can be used against side-channel attack (SCA). It was shown that the distance of a code with linear complementary dual represents the security degree against SCA. Hence, for application in SCA, larger minimum distance is preferred. MDS codes form an optimal family of classical codes in the sense that the minimum distance achieves the Singleton bound.

In \cite{Mass92}, Massey showed that there exist asymptotically good LCD codes by establishing a relationship between LCD codes and linear codes and raised a question on whether LCD codes can achieve the Gilbert-Varshamov bound. Later, Sendrier showed that LCD codes meet the asymptotic Gilbert-Varshamov bound using the hull dimension spectra of linear codes \cite{Send04}.  In the literature, LCD codes have been studied extensively and many results and properties for LCD cyclic codes were found \cite{DKOSS,EY09,Ding,ML86,TH70,MS,YM94}. In \cite{Jin17}, the problem of constructing LCD MDS codes over finite fields of even characteristic  was solved. For the case where the field characteristic is odd, several results were presented in  \cite{Chen, Jin17,MS,SK}.
After this work was completed, we became aware that there is a contemporaneous paper on LCD MDS codes by Carlet et al. \cite{CMTQ17}. A few days later, it was further showed by them that any linear code is equivalent to an LCD code \cite{CMTQ172}. As a result, the problem of classifying LCD MDS codes was settled completely in \cite{CMTQ17, CMTQ172}. While the landmark result in \cite{CMTQ17} cleverly employs techniques from linear algebra, the current work uses the language of function fields and algebraic geometry codes. This enables us to describe and obtain some LCD MDS codes in a very different way. Because of this, we feel that the current work is still of independent interest and may facilitate future investigations of (some) LCD MDS codes, since the full machinery of function fields is available for these codes. We explicitly construct several classes of LCD MDS codes over finite fields of odd characteristic. The algebraic geometry framework gives rise to a finer analysis as to when a generalized Reed-Solomon code is an LCD code as well. Summarizing, we obtain the following  result.
\begin{theorem}\label{thm:main} Let $q=p^r$ for an odd prime power $q$. Then there exist explicit $q$-ary $[n,k]$-LCD MDS codes for the following ranges of $n$ and $k$.
\begin{itemize}
  \item [(1)] Set $n=p^t$, with $1\le t < r$ and $0\le k\le p^t$ (see Construction \ref {ex:3.6} (i)).
  \item [(2)] Set $n<q$ a positive integer such that $p|n$ and $(n-1)|(q-1)$ and $0\le k \le n$ (see Construction \ref{ex:3.6} (ii)).
  \item [(3)] Set $n=p \ell$, with $\ell$ a divisor of $q-1$ such that $p\ell<q$. If $\ell$ is even, then $k=2\ell-2m$ for any $m$ satisfying $0\le m\le\ell-1$. If $\ell$ is odd, then $k=2\ell-2m$ for any $m$ such that $0 \le m \le \ell$ (see Construction \ref{ex:3.7}).
  \item[(4)] Suppose that $q\ge 5$ is an odd square and let $n<q$ such that $n-1$ is an odd divisor of $q-1$. Then $k=n-2m$ for any $m$ such that $0 \le m\le n/2$ (see Construction \ref{ex:3.8}).
  \item[(5)] Set $n$ to be a divisor of $(q-1)/2$. If $n$ is even (resp.~odd) $k=n-2m$ (resp. $k=n-2m-1$) for any $m$ such that $1 \le m \le n/2$ (see Construction \ref{ex:3.9}).
  \item[(6)] Let $r>1$ and write $N_t=(p^t-1)/(p-1)$. Then we set $n=pN_{r-1}$ and $k$ any integer such that $0 \le k \le pN_{r-1}$ (see Construction \ref{ex:3.10}).
  \item[(7)] Let $t<r$ and choose $1 \le d \le p^{r-t}$ such that $\gcd(d,p)=1$. Then we can set $n=dp^t$ and $k=dp^t-2m$ for any $m$ such that $2 \le m \le p^t/2$ (see Construction \ref{ex:3.11}).
\end{itemize}
\end{theorem}

\begin{remark}{\rm
Note that part (1)  of Theorem \ref{thm:main} was obtained in \cite[Theorem 3.5]{Chen} using generalized Reed-Solomon codes as well. The parameters in part (5) of Theorem \ref{thm:main} were also obtained in \cite[Theorem 9]{SK}, though using a different construction involving negacyclic codes. Finally, as mentioned before, in \cite{CMTQ17} linear algebra techniques are used to construct LCD MDS codes for all feasible parameters, hence also the above. 
}
\end{remark}
The paper is organized as follows. We first provide the relevant backgrounds on rational function fields and algebraic geometry codes in Section II. In Section III, we consider the construction of LCD MDS codes from algebraic geometry codes over the rational function field. Sufficient conditions are given for two algebraic geometry codes to be disjoint. This gives rise to Theorem \ref{thm:3.3} describing how to construct LCD MDS codes. We then give several constructions that provide the parameters mentioned in Theorem \ref{thm:main}.

\section{Preliminaries}
\subsection{LCD MDS codes}
Throughout this paper, denote by $\F_q$  the finite field of $q$ elements with characteristic $p$. We define the Euclidean inner product of two vectors $\bu=(u_1,\dots,u_n)$ and $\bv=(v_1,\dots,v_n)$ of $\F_q^n$ by $\bu\cdot\bv=\sum_{i=1}^nu_iv_i$. A $q$-ary linear code $C$ of length $n$ is a subspace of $\F_q^n$. The dual code of $C$ is defined by $C^{\perp}=\{\bx\in\F_q^n:\; \bx\cdot\bc=0\; \mbox{for all $\bc\in C$}\}$. $C$ is called  {\it linear complementary dual} (LCD for short) if $C\cap C^{\perp}=\{\bo\}$.

A linear code with length $n$, dimension $k$ and minimum distance $d$ is called {\it maximum distance separable} (MDS for short) if $k+d=n+1$. In this article we study codes that are both LCD and MDS, that is to say, LCD MDS codes.

\subsection{Rational function field}
In this subsection, we briefly review some basic results on function field (in particular, on rational function field). For the details on the theory of function field, the reader may refer to the book \cite{Stich08}.

 Denote by $\mF$ the rational function field $\F_q(x)$ with a transcendental element $x$ over $\F_q$. For an element $\alpha \in \F_q$, the zero place of $x-\alpha$ is denoted by $P_\alpha$ and  its pole place by $P_\infty$. These are all rational places (or places of degree $1$). Furthermore, every place $R\neq P_\infty$ corresponds to a monic irreducible polynomial $r(x)\in\F_q[x]$. The degree of $R$, denoted by $\deg(R)$, is defined to be the degree of $r(x)$. Let $\PP_\mF$ denote the set of places of $\mF$.

A divisor $G$ of $\mF$ is a formal sum $\sum_{R\in\PP_\mF}m_RR$ with only finitely many nonzero $m_R$. The support of $G$ is defined to be $\{R\in\PP_\mF:\; m_R\neq 0\}$.  The degree of $G$ is defined to be $\sum_{R\in\PP_\mF}m_R\deg(R)$. A divisor $G=\sum_{R\in\PP_\mF} m_RR$ is said to be bigger than or equal to other divisor $D=\sum_{R\in\PP_\mF} n_RR$ if $m_R\ge n_R$ for all $R\in\PP_\mF$. A divisor $G=\sum_{R\in\PP_\mF}m_RR$ is said to be effective, denoted by $G\ge 0$ if $m_P\ge 0$ for all $R\in\PP_\mF$. For two divisors $G=\sum_{R\in\PP_\mF} m_RR$ and $D=\sum_{R\in\PP_\mF} n_RR$, we define
\[G\vee D:=\sum_{R\in\PP_\mF} \max\{m_R,n_R\}R,\qquad G\wedge D:=\sum_{R\in\PP_\mF} \min\{m_R,n_R\}R.\]
It is clear that
\begin{equation*}
\deg(G\wedge D) +  \deg(G\vee D)= \deg(G)+\deg(D).
\end{equation*}

Assume that a nonzero polynomial $f(x)\in\F_q[x]$ has the canonical factorization $\Ga\prod_{i=1}^t r_i(x)^{e_i}$ with $\Ga\in\F_q^*$ and pairwise distinct monic irreducible polynomials $r_i(s)$, the principal  divisor $(f)$ of $f(x)$ is $\sum_{i=1}^te_i\deg(R_i)-\deg(f)P_\infty$, where $R_i$ are places corresponding to $r_i(x)$. Now for a rational function $f(x)/g(x)\in\mF$ with $g(x)\neq 0$ and $f(x)\neq0$, the principal divisor $(f/g)$ of $f(x)/g(x)$ is defined to be $(f)-(g)$. For a nonzero function $u\in\mF$, we write $(u)=\sum_{P\in \tP}m_PP-\sum_{R\in\tR}m_RR$, where $\tP$ and $\tR$ are two disjoint subsets of $\PP_\mF$ and $m_P>0$, $m_R> 0$ for all $P\in\tP$ and $R\in\tR$. The divisors $\sum_{P\in \tP}m_PP$ (denoted by $(u)_0$) and $\sum_{R\in\tR}m_RR$ (denoted by $(u)_\infty$) are called zero divisor and pole divisor of $(u)$, respectively. It is well known that $\deg((u)_0)=\deg((u)_{\infty})$. In particular any principal divisor has degree zero. For the rational function field $\mF$ the converse holds: for any divisor $G$ of degree zero, one can find a function $y\in\mF$ such that $(y)=G$.

For a divisor $G$, we define the Riemann-Roch space
\[\mL(G):=\{u\in\mF\setminus\{0\}:\; (u)+G\ge 0\}\cup\{0\}.\]
Then $\mL(G)$ is an $\F_q$-subspace of dimension $\deg(G)+1$ for any divisor of nonnegative degree. If for example $G=mP_\infty$, then $\mL(G)$ is the $(m+1)$-dimensional space of polynomials of degree at most $m$.
It is straightforward to verify that
\begin{equation*}
\mL(G) \cap \mL(H) = \mL(G\wedge H) \ \makebox{and} \ \mL(G)+\mL(H) \subseteq \mL(G\vee H).
\end{equation*}


\subsection{Algebraic geometry codes on the rational function field}
Let us first define generalized Reed-Solomon codes. Let $\ba=(\Ga_1,\dots,\Ga_n)$ with $\Ga_1,\dots,\Ga_n$ being $n$ distinct elements of $\F_q$. Let $\bv=(v_1,\dots,v_n)\in(\F_q^*)^n$ be a vector. For $1\le k\le n$, we define
\begin{equation*}
\RS_k(\ba,\bv):=\{(v_1f(\Ga_1),\dots,v_nf(\Ga_n)):\; f(x)\in\F_q[x],\; \deg(f)\le k-1\}.
\end{equation*}
$\RS_k(\ba,\bv)$ is called a generalized Reed-Solomon code. It is an $[n,k]$-MDS code over $\F_q$.

For convenience, we write $P_i:=P_{\Ga_i}$ and let $D$ be the divisor $\sum_{i=1}^nP_{i}$. Further let $G$ be a divisor such that $\Supp(D)\cap\Supp(G)=\emptyset$. Define the following functional  algebraic geometry code
\begin{equation*}
C_L(D,G):=\{(f(P_{1}),\dots,f(P_{n})):\; f\in\mL(G)\}.
\end{equation*}
It is easy to see that $C_L(D,(k-1)P_\infty)$ is the same as $\RS_k(\ba,\bv)$ with $\bv=(1,\dots,1)$. More generally \cite[Proposition 2.3.3]{Stich08} implies the following lemma.
\begin{lemma}\label{lem:2.1a} If $1\le k\le n$ and $G$ is a divisor such that $\deg(G)=k-1$ and $\Supp(G) \cap \Supp(D) =\emptyset$, then
$C_L(D,G)$ is equal to the generalized Reed-Solomon code $\RS_k(\ba,\bu)$ for some $\bu\in(\F_q^*)^n$. Conversely,  any generalized Reed-Solomon code $\RS_k(\ba,\bv)$ can be realized as a functional code $C_L(D,H)$ for some divisor $H$ of degree $k-1$ with $\Supp(H) \cap \Supp(D) =\emptyset$.
\end{lemma}

To define the dual code of an algebraic geometry code, we need to introduce differentials. For a place $R$ corresponding to a monic irreducible polynomial $r(x)\in\F_q[x]$ and a nonzero polynomial $f(x)\in\F_q[x]$, we denote $k$ by $\nu_R(f)$, where $r(x)^k||f(x)$. This evaluation can be extended to any nonzero function $f(x)/g(x)\in\mF$ with $f(x),g(x)\in\F_q[x]$ by defining $\nu_R(f/g)=\nu_R(f)-\nu_R(g)$. In particular, we set $\nu_R(0)=+\infty$. If $R$ is the place $P_\infty$, we define $\nu_{P_\infty}(f/g)=\deg(g)-\deg(f)$.

The differential space of $\mF$ is defined to be
\[\Omega_\mF:=\{fdx:f\in\mF\}.\]
This is a one-dimensional space over $\mF$. For any place $R\neq P_\infty$, we define $\nu_R(fdx)=\nu_R(f)$. Furthermore, we define $\nu_{P_\infty}(fdx)=\nu_{P_{\infty}}(f)-2$. For a nonzero function $f$, the divisor $(fdx)=\sum_{R\in\PP_\mF}\nu_R(fdx)R$ is called a canonical divisor. It is clear that a canonical divisor has degree $-2$.

For a divisor $G$, we define the space
\begin{equation*}
\Omega(G):=\{\omega\in\Omega_\mF\setminus\{0\}:\; (\Go)\ge G\}\cup\{0\}.
\end{equation*}
Then $\Omega(G)$ is an $\F_q$-subspace of  $\Omega_\mF$ of dimension $-\deg(G)-1$ if $\deg(G)\le -2$. For an element $\Ga\in\F_q$ and a function $f$ with $\nu_{P_\Ga}(f)\ge -1$, we can write $f=a_{-1}/(x-\Ga)+a_0+a_1(x-\Ga)+\dots$. The residue of $fdx$, denoted by $\res_{P_\Ga}(fdx)$, is thus defined to be $a_{-1}$.

As before, let $D=\sum_{i=1}^nP_{i}$. For a divisor $G$ with $\Supp(D)\cap\Supp(G)=\emptyset$, the differential algebraic geometry code is defined as  follow
\begin{equation*}
C_\Omega(D,G):=\{(\res_{P_{1}}(\Go),\dots, \res_{P_{n}}(\Go)):\; \Go\in\GO(G-D)\}.
\end{equation*}
The following lemma can be found in \cite[Theorem 2.2.8 and Proposition 2.2.10]{Stich08}.

\begin{lemma}\label{lem:2.3}
Let $D=\sum_{i=1}^nP_{\Ga_i}$. For a divisor $G$ with $\Supp(D)\cap\Supp(G)=\emptyset$, we have the following.
\begin{itemize}
\item[{\rm (i)}] The dual code of $C_L(D,G)$ is $C_\Omega(D,G)$.
\item[{\rm (ii)}] If there exists a differential $\eta$ such that $\nu_{P_{i}}(\eta)=-1$ and $\res_{P_{\Ga_i}}(\eta)=1$ for all $1\le i\le n$. Then $C_\Omega(D,G)=C_L(D,D-G+(\eta))$.
\end{itemize}
\end{lemma}

In \cite[Lem. 2.3.6]{Stich08} a differential satisfying the conditions in \ref{lem:2.3}(ii) is constructed explicitly. Before stating this construction, it is convenient to define
\begin{equation*}
g:=\prod_{i=1}^n(x-\Ga_i) \ \makebox{and} \ z:=\frac{dg}{dx}=\sum_{i=1}^n \prod_{j=1:j\neq i}^n(x-\Ga_j).
\end{equation*}
A first property of the function $z$ is given in the following lemma.
\begin{lemma}\label{lem:2.1}
Let $n\le q$ and let $\Ga_1,\dots,\Ga_n$ be distinct elements of $\F_q$ and define $z$ as above.
Then
$$\deg z =  \left\{
\begin{array}{rl}
n-1  & \makebox{if $p \not| n$,}\\
n-2  & \makebox{if $p | n$ and $\sum_{i=1}^n \alpha_i \neq 0.$}\\
<n-2 & \makebox{otherwise}
\end{array}
\right.$$
Moreover, $z(\Ga_i) \neq 0$ for all $i$.
\end{lemma}
\begin{proof}
Since $z=z(x)=\frac{dg}{dx}$, we see that $z(x)=nz(x)^{n-1}-(n-1)(\sum_i \Ga_i) z(x)^{n-2}+\cdots$. The first part of the lemma now follows. The final statement follows, since $\Ga_i$ is a simple root of $g$ for all $i$ and hence not a zero of $z$.
\end{proof}

With this notation, the construction in \cite[Lem. 2.3.6]{Stich08} can be paraphrased as follows.
\begin{lemma}\label{lem:2.4}
Let $\alpha_1,\dots,\alpha_n$ be distinct elements of $\F_q$, and let $z$ be as in Lemma \ref{lem:2.1}. Then the differential $$\omega_z:=\left(\sum_{i=1}^n \frac{1}{x-\alpha_i}\right)dx=\frac{z}{(x-\alpha_1)\cdots(x-\alpha_n)}dx$$ has divisor
$$W_z:= (z) -D +(n-2)P_{\infty}.$$
Moreover $\res_{P_i}(\omega_z)=1$ for all $i$ between $1$ and $n$.
\end{lemma}
It will be convenient to write $Z$ for the divisor of zeroes of $z$. With this notation, we have $(z)=Z-(\deg(z)) P_\infty.$ Lemma \ref{lem:2.1} implies that $\deg(Z)=\deg(z) \le n-1$ as well as that $\Supp(Z) \cap \Supp(D) = \emptyset.$
Combining Lemmas \ref{lem:2.3} and \ref{lem:2.4}, we obtain the following result.

\begin{lemma}\label{lem:2.5}
Let $D$, $G$ be as defined in Lemma \ref{lem:2.4}. Then the dual of $C_L(D,G)$ is $C_L(D,(z)-G+(n-2)P_\infty)$.
\end{lemma}
Using the divisor $Z$, we may also write that $C_L^\perp(D,G)=C_L(D,Z-G+(n-2-\deg(z))P_\infty)$.


\section{Construction of LCD MDS codes}
In this section we construct several classes of LCD MDS codes. We first give a sufficient condition under which two functional algebraic geometry codes are disjoint.

\begin{lemma}\label{lem:3.2}
Assume that $A,B$ are two positive divisors such that
 \begin{itemize}
   \item [(i)] $\Supp(A)$, $\Supp(B)$, $\Supp(D)$ and $\Supp(H)$ are pairwise disjoint;
   \item [(ii)] $\deg(D)=n>\deg(H)$ and $\deg(A)+\deg(B)>\deg(H)$.
 \end{itemize}
Further let $w \in \mathbb{F}_q(x)$ be a function satisfying $v_{P_{\Ga_i}}(w)=0$ for all $i$ between $1$ and $n$. Then the codes $C_L(D,H-A+(w))$ and $C_L(D,H-B+(w))$ are disjoint.
 \end{lemma}
 \begin{proof}
Write $v_i=w^{-1}(P_{i})$. Then for all $i$ we have $v_i \neq 0$. Now assume that $\bc \in C_L(D,H-A+(w)) \cap C_L(D,H-B+(w)).$ Then there exist $f_1\in \mL(H-A)$ and $f_2\in \mL(H-B)$ such that $(v_1f_1(P_1),\dots,v_nf_1(P_n))=\bc=(v_1f_2(P_1),\dots,v_nf_2(P_n)).$ Since for all $i$ we have $v_i \neq 0$ and $\Supp(D) \cap \Supp(H)=\emptyset$, we have $f_1-f_2\in L(H-D)$. Hence $f_1-f_2=0$, since $\deg(H)<n$. Thus, $f_1=f_2\in \mL(H-A) \cap \mL(H-B)= \mL((H-A)\wedge (H-B))=\mL(H-A-B)=\{0\}$. Here we used the assumption that $\deg(A)+\deg(B)>\deg(H)$ in the last equality.
 \end{proof}

Now we are going to construct LCD MDS codes. The idea is to use Lemma \ref{lem:3.2} for suitably chosen divisors $H$, $A$ and $B$ and a function $y$ such that $C_L^{\perp}(D,H-A+(y))=C_L(D,H-B+(y))$.

\begin{theorem}\label{thm:3.3}
Let $D=\sum_{i=1}^nP_i$, $H$ a divisor and $A,B$ two positive divisors such that:
\begin{enumerate}
 \item[(i)] $\Supp(A)$, $\Supp(B)$, $\Supp(D)$ and $\Supp(H)$ are pairwise disjoint;
 \item[(ii)] $\deg(H)=n-1$ and
 \item[(iii)] $2H-A-B-(z)-(n-2)P_{\infty}$ equals $(y)$ for some element $y\in \F_q(x)$ such that $y(P_{i})$ are squares in $\F_q^*$ for all $1\le i\le n$.
\end{enumerate}
Then for any $w \in \F_q(x)$ such that $w(P_i)^{-2}=y(P_i)$ for all $1\le i\le n$, the code $C_L(D,H-A+(w))$ is an $[n,n-\deg(A)]$ LCD MDS code.
\end{theorem}
\begin{proof} First of all, note that the assumptions (ii) and (iii) imply that $\deg(A)+\deg(B)=n$, since any principal divisor has degree zero. Further note that the support of $2H-A-B-(z)-(n-2)P_{\infty}$ is disjoint with $\Supp(D)$.

By Lemma \ref{lem:3.2}, we know that $ C_L(D,H-A+(w))$ and $C_L(D,H-B+(w))$ are disjoint. It is now sufficient to show that the dual of $ C_L(D,H-A+(w))$ is exactly $C_L(D,H-B+(w))$. By Lemmas \ref{lem:2.5}, the dual of $C_L(D,H-A+(w))$ is \[C_L(D,-H+A-(w)+(z)+(n-2)P_{\infty})=C_L(D,H-B-(w)-(y))= C_L(D,H-B+(w)-(yw^2)).\]

We claim that $C_L(D,H-B+(w)-(yw^2))=C_L(D,H-B+(w))$. There is a natural isomorphism of vector spaces between the Riemann--Roch spaces $L(H-B+(w)-(yw^2))$ and $L(H-B+(w))$ sending $f \in L(H-B+(w)-(yw^2))$ to $f/(yw^2) \in L(H-B+(w))$. This map induces an isomorphism of codes $\phi: C_L(D,H-B+(w)-(yw^2)) \rightarrow C_L(D,H-B+(w))$ defined by
$$\phi(f(P_1),\dots,f(P_n))=\left(\frac{f}{yw^2}\left(P_1\right),\dots,\frac{f}{yw^2}\left(P_n\right)\right).$$
However, since $(yw^2)(P_i)=1$ for all $1\le i\le n$, we have $\phi(f(P_1),\dots,f(P_n))=(f(P_1),\dots,f(P_n))$. Hence $C_L(D,H-B+(w)-(yw^2))=C_L(D,H-B+(w))$ as claimed.
This completes the proof.
\end{proof}

Note that by the assumption that $y(P_i)$ is a non-zero square for all $i$, an element $w$ satisfying $w(P_i)^{-2}=y(P_i)$ exists. Moreover, the code $C_L(D,H-A+(w))$ does not depend on the choice of $w$.

By setting $H=(n-1)P_\infty$ in Theorem \ref{thm:3.3}, we obtain the main result of \cite{Jin17} as a corollary.
\begin{corollary}\label{cor:3.4}
Let $a(x),b(x)$ be two co-prime polynomials. Let $A, B$ be the zero divisors of $a(x)$ and $b(x)$, respectively.  Assume that
\begin{enumerate}
\item[(i)] $\Supp(D)$ is disjoint with both $\Supp(A)$ and $\Supp(B)$;
\item[(ii)] $\deg(a(x))+\deg(b(x))=n$ and
\item[(iii)] $(abz)(P_{i})$ are square elements of $\F_q^*$ for all $1\le i\le n$.
\end{enumerate}
Then for any element $w \in \F_q(x)$ such that $w(P_i)^{2}=(abz)(P_i)$ for all $1\le i\le n$, the code $C_L(D,(n-1)P_\infty-A+(w))$ is an LCD MDS code.
 \end{corollary}
\begin{proof} In Theorem \ref{thm:3.3}, set $H=(n-1)P_\infty$.  Then $2H-A-B-(z)-(n-2)P_{\infty}$ is the principal divisor of $1/(abz)$. The desired result now follows from Theorem \ref{thm:3.3}.
\end{proof}

One way to make sure that the function $y$ in Theorem \ref{thm:3.3} satisfies that $y(P_i)$ is a non-zero square for all $i$ is by making sure that $y$ itself is a square. This is the main idea behind the following corollary.

\begin{corollary}\label{cor:3.5} Let $m\ge 0$ and $3 \le n \le q$. Further, let $X$ and $Y$ be positive divisors such that $Z=2X+Y$. Finally let $Y_1$ and $Y_2$ be disjoint positive divisors such that $Y=Y_1+Y_2$.
 \begin{itemize}
\item[{\rm (i)}]  There exists a $q$-ary $[n, n-2m-\deg(Y_1)]$-LCD MDS code for any $2\le m\le (n-\deg(Y))/2$ or $m=0$;
\item[{\rm (ii)}]  If there exists a place $P$ of degree one not in $\Supp(D) \cup \Supp(Y_2) \cup \{P_\infty\}$ and $\deg(Y) \le n-2$, then there exists a $q$-ary $[n, n-2-\deg(Y_1)]$-LCD MDS code.
\end{itemize}
\end{corollary}
  \begin{proof} First we prove {\rm (i)}. Choose $Q$ to be a place of degree $n-1$ and set $H=Q$. Since $n \ge 3$, we have $\Supp(H)\cap\Supp(D)=\emptyset$. If $m=0$, set $K=0$, which trivially implies that $\Supp(K)\cap\Supp(D)=\emptyset$, $\Supp(K)\cap\Supp(H)=\emptyset$ and $\Supp(K) \cap \Supp(Y_2)=\emptyset$. If $m \ge 2$, we set $K=R$, with $R$ a place of degree $m$. Since $m \neq 1$, we have $\Supp(K)\cap\Supp(D)=\emptyset$. Moreover, since $n \ge 3$, we have $m \le n/2 < n-1$, whence $\Supp(K)\cap\Supp(H)=\emptyset$. Next, we show that we can choose $R$ such that $R \not\in \Supp(Y_2).$
The set $\Supp(Y_2)$ contains at most $\deg(Y_2)/m$ places of degree $m$. Note that $$\frac{\deg(Y_2)}{m} \le \frac{\deg(Z)}{m} \le \frac{n-1}{m} \le \frac{q-1}{m}.$$ On the other hand, the number of places of degree $m$ is at least $(q^m-q^{\lfloor m/2 \rfloor+1})/m$ for $m>2$ and equal to $(q^2-q)/2$ for $m=2$. Hence, we see that we can choose the place $R$ of degree $m$ such that $R \not\in \Supp(Y_2).$

Now with such a choice of $R$, let $A=2K+Y_1$ and $B=Y_2+(n-2m-\deg (Y))P_{\infty}$. Then $$2H-A-B-(z)-(n-2)P_{\infty}=2(H-K-X-Y-(n-1-m-\deg(X)-\deg(Y))P_{\infty})$$ is equal to $(f^2)$ for some $f\in\F_q(x)$. The desired result follows from Theorem \ref{thm:3.3} since  $y(P_i)=f(P_i)^2$ are non-zero square elements of $\F_q$ for all $1\le i\le n$.

The proof of {\rm (ii)} is similar. We choose $K=P$.
\end{proof}

Note that if $n<q$ and $Y_2$ is chosen to be $0$, the place $P$ in part (ii) of the theorem is guaranteed to exist.

\begin{remark}{\rm
In characteristic two, the derivative of any polynomial is a square. In particular, the function $z$ is always a square, which implies that $Z=2X+Y$ for a positive divisor $X$ and $Y=0$, whence we can choose $Y_1=Y_2=0$. Corollary \ref{cor:3.5} then implies that in characteristic two, for any $0 \le m \le n/2$ and $3 \le n \le q$ there exists an $[n,n-2m]$-LCD MDS code. This result is already contained in \cite{Jin17}, where $[n,k]$-LCD MDS codes were found in even characteristic for any $0 \le n \le q$ and $0 \le k \le n$.}
\end{remark}

\begin{remark}{\rm
Suppose $q$ is odd and $n=q$. Then $g=x^q-x$ in our construction, implying that $Z=0$. Using Corollary \ref{cor:3.5} with $X=Y=Y_1=Y_2=0$, we obtain $[q,q-2m]$-LCD MDS codes for $m=0$ and $2 \le m \le q/2$. Considering the duals of these codes, we obtain $[q,2m]$ codes for $2 \le m \le q/2$. Combined, we see that one can obtain explicit $[q,k]$-LCD MDS codes for nearly all values of $k$ between $0$ and $q$, except $k \in \{2,q-2\}$.}
\end{remark}

Because of the above two remarks, our main interest is to construct explicit $[n,k]$-LCD MDS codes in case $q$ is odd and $n<q.$ We now give several explicit constructions, exemplifying the versatility of Theorem \ref{thm:3.3} and its two corollaries.

\begin{construction}\label{ex:3.6}{\rm  Let $q=p^r$ with a prime $p$ and integer $r\ge 1$.
   \begin{itemize}
\item[{\rm (i)}]  For any integer $t$ with $1\le t\le r$, let $V$ be an $\F_p$-subspace of $\F_q$ of dimension $t$. Put $g=\prod_{\Ga\in V}(x-\Ga)$. Then $z=\frac{dg}{dx}=\prod_{\Ga\in V\setminus\{0\}}\Ga$ is a nonzero constant. Hence we have $Z=0$. By Corollary \ref{cor:3.5}, there exists a $q$-ary $[p^t, p^t-2m]$-LCD MDS code for $m=0$ (which is trivial in this case) or $2\le m\le p^t/2$. If $t<r$, such a code exists for $0\le m\le p^t/2$. Considering dual codes as well, we see that if $p$ is odd and $t<r$, we can construct explicit $[p^t,k]$-LCD MDS codes for any $k$ satisfying $0 \le k \le p^t$.

  \item[{\rm (ii)}] Let $n$ be a positive integer with $p|n$ and $(n-1)|(q-1)$. Put $g(x)=x^n-x.$ Then $g(x)$ has all roots in $\F_q$ and every root is simple. Furthermore, $z=\frac{dg}{dx}=-1$, whence $Z=X=Y=0$. Using Corollary \ref{cor:3.5}, we find an explicit $q$-ary $[n, n-2m]$-LCD MDS code for $m=0$ and $2\le m\le n/2$. If $n<q$ and $n$ is odd, we obtain explicit $[n,k]$-LCD MDS codes for any $k$ with $0 \le k \le n$ by considering the previous constructed codes and their duals.

      For example, $q=81$ and $n=21$. Then one obtains explicit $81$-ary $[21, k]$-LCD MDS codes for any $0\le k\le 21$.
   \end{itemize}
}\end{construction}

  \begin{construction}\label{ex:3.7}{\rm  Let $q=p^r$ with a prime $p$ and integer $r\ge 1$.
  Let $q-1=\ell\cdot d$ and suppose that $p < d$. Further, let $U_{\ell}$ be the multiplicative subgroup of $\F_q^*$ of order $\ell$ and $\beta_1U_{\ell},\dots,\beta_p U_{\ell}$ be pairwise distinct cosets. Consider the polynomials $f(x):=\prod_{i=1}^p(x-\beta_i^\ell)$ and $g(x):=f(x^\ell)$. Then $g(x)$ has no multiple roots and $z=\frac{dg}{dx}=\ell x^{\ell-1}\frac{df}{dx}(x^\ell)$. Choosing $\beta_i=\gamma^i,$ with $\gamma$ a primitive element of $\F_q^*$, we have $\sum_{i=1}^p \beta_i^\ell \neq 0$ and $\frac{df}{dx}(P_0) \neq 0$. Hence in this case $\deg(z)=\ell-1+\ell(p-2)=\ell(p-1)-1$ and $Z=(\ell-1)P_0+E$ for some positive divisor $E$ of degree $\deg(E)=\ell(p-2)$ with $P_0 \not\in \Supp(E)$.

  If $\ell$ is even, we can choose $X=(\ell-2)/2P_0$ and $Y=P_0+E$. Choosing $Y_1=Y$ and $Y_2=0$, we obtain from Corollary \ref{cor:3.5} explicit $[p\ell,2\ell-1-2m]$-LCD MDS codes for $0 \le m \le \ell-1$. Choosing $Y_1=E$ and $Y_2=P_0$, we obtain from Corollary \ref{cor:3.5} explicit $[p\ell,2\ell-2m]$-LCD MDS codes for $0 \le m \le \ell-1$. Note that the assumption $p<d$ implies that $p\ell \le (d-1)\ell=q-1-\ell \le q-2$. Therefore a place $P$ as in part (ii) of Corollary \ref{cor:3.5} exists.

  If $\ell$ is odd, we can choose $X=(\ell-1)/2P_0$ and $Y=E$.  Choosing $Y_1=E$ and $Y_2=0$, we construct $[p\ell, 2\ell-2m]$-LCD MDS codes for $0 \le m \le \ell$.
   }\end{construction}

  \begin{construction}\label{ex:3.8}{\rm Let $q\ge 5$ be an odd square. Let $n-1$ be an odd divisor of $q-1$ and $n<q$. Then $n-1\le (q-1)/2$, i.e., $n\le (q-1)/2+1$.  Let $g(x)=x^n-x$. Then $g(x)$ has no multiple roots and splits completely in $\F_q$. Furthermore, $z(x)=\frac{dg}{dx}=nx^{n-1}-1$. Label  the roots of $g(x)$ by $\Ga_1,\dots,\Ga_n$. Then we have $z(\Ga_i)=n-1$ or $-1$ for all $1\le i\le n$. Since $q$ is a square, we can find $\theta,\delta\in\F_q$ such that $\theta^2=n-1$ and $\delta^2=-1$.

  Now for any $1\le m\le n/2$, choose two distinct  elements $\Gb,\Gg\in \F_q$ that are not  roots of $g(x)$ (this is possible since $n\le (q-1)/2+1\le q-2$). Put  $a(x)=(x-\Gb)^{2m}$ and  $b(x)=(x-\Gg)^{n-2m}$. Then $(abz)(\Ga)$ is equal to $(\theta(\Ga-\Gb)^m(\Ga-\Gg)^{n/2-m})^2$ or $(\delta(\Ga-\Gb)^m(\Ga-\Gg)^{n/2-m})^2$ for all roots $\Ga$ of $g(x)$. Thus, by Corollary \ref{cor:3.4} we obtain a $q$-ary $[n,n-2m]$-LCD MDS code  for all $m\le n/2$.
   }\end{construction}

 \begin{construction}\label{ex:3.9}{\rm  Let $q$ be odd and let $n$ be a divisor of $(q-1)/2$. Let $g(x)=x^n-1$. Then every root of $g(x)$ is a square element of $\F_q$. It is clear that $g(x)$ has no multiple roots and splits completely in $\F_q$. Furthermore, $z(x)=\frac{dg}{dx}=nx^{n-1}$.

For any $1\le m\le n/2$, choose two distinct  elements $\Gb,\Gg\in \F_q$ that are not  roots of $g(x)$ (this is possible since $n\le  q-2$). Put
 \begin{eqnarray*}
 (a(x),b(x))=\left\{\begin{array}{ll}
(n(x-\Gb)^{2m},(x-\Gg)^{n-2m})&\mbox{if $n$ is even}\\
(nx(x-\Gb)^{2m},(x-\Gg)^{n-1-2m})&\mbox{if $n$ is odd}
\end{array}
 \right.
 \end{eqnarray*}
 For every root $\Ga$ of $g(x)$, let $\Ga=\Ga_1^2$ for some $\Ga_1\in\F_q$.
  Then
   \begin{eqnarray*}
 (abz)(\Ga)=\left\{\begin{array}{ll}
(n\Ga_1^{n-1}(\Ga-\Gb)^m(\Ga-\Gg)^{n/2-m})^2&\mbox{if $n$ is even}\\
(n\Ga_1^{n}(\Ga-\Gb)^m(\Ga-\Gg)^{(n-1)/2-m})^2&\mbox{if $n$ is odd}
\end{array}
 \right.
 \end{eqnarray*}
   Thus, by Corollary \ref{cor:3.4} we obtain a $q$-ary $[n,n-2m]$-LCD MDS code  for even $n$ and a $q$-ary $[n,n-2m-1]$-LCD MDS code for odd $n$.
   }\end{construction}

\begin{construction}\label{ex:3.10}
{\rm
Let $q=p^r$ for some odd prime power $p$ and integer $r>1$. Define $N_r:=(p^r-1)/(p-1).$ Further define the polynomial $g=( (x+1)^{N_r}-1 )/x$. Then $g$ is a polynomial of degree $n:=N_r-1=p N_{r-1}$ with simple roots, all in $\F_q$. A direct computation shows that $z=-x^{p-2}\left(((x+1)^{N_{r-1}}-1)/x\right)^p,$ implying that $\deg(z)=n-2$. Moreover, we see that $Z=(p-2)P_0+p\sum_{i=1}^s Q_i$, where $Q_1,\dots,Q_s$ denote the zeroes of $(x+1)^{N_{r-1}}-1$ different from $P_0.$ Note that $z$ has all its roots in $\F_{p^{r-1}}$. Since $q$ is odd, we can choose $Y=P_0+\sum_{i=1}^s Q_i$, which is a divisor of degree
$$\deg(Y)=1+\sum_{i=1}^s \deg(Q_i)=1+\frac{\deg(Z)-p+2}{p}=1+\frac{n-p}{p}=N_{r-1}.$$
Choosing $Y_1=0$, we obtain using Corollary \ref{cor:3.5} a $q$-ary $[p N_{r-1},p N_{r-1}-2m]$-LCD MDS code for $0 \le m \le (p^{r-1}-1)/2.$ Note that since $z$ has its roots in $\F_{p^{r-1}}$, there will be a place of degree one, satisfying the conditions in part (ii) of  Corollary \ref{cor:3.5}. Choosing $Y_1=P_0$, we similarly obtain a $q$-ary $[p N_{r-1},p N_{r-1}-2m-1]$-LCD MDS code for $0 \le m \le (p^{r-1}-1)/2.$

Combining these two results, we see that we have obtained an explicit $q$-ary $[p N_{r-1},k]$-LCD MDS code for $N_{r-1}-1 \le k \le p N_{r-1}.$ Considering the duals of these codes, we can strengthen this conclusion to the statement that we obtain an explicit $q$-ary $[p N_{r-1},k]$-LCD MDS code for $0 \le k \le p N_{r-1}.$
}
\end{construction}

\begin{construction}\label{ex:3.11}
{\rm
Suppose $q=p^r$ and let $V \subset \F_q$ be an $\F_p$ vector space of dimension $t < r$. We define $g_V(x):=\prod_{\alpha \in V}(x-\alpha)$. Further let $\alpha_1+V,\dots,\alpha_d+V$ be mutually distinct cosets of $V$ in $\F_q$. This property can also be expressed by saying that $g_V(\alpha_1),\dots,g_V(\alpha_d)$ are mutually distinct. It is clear that $1 \le d \le p^{r-t}.$ Now define $f:=\prod_{i=1}^d(x-g_V(\alpha_i))$ and set $g:=f(g_V(x))$. The polynomial $g$ has no multiple roots and its roots are exactly the $dp^t$ elements of $\F_q$ occurring in the cosets $\alpha_1+V,\dots,\alpha_d+V$.

We see that $z=\frac{df}{dx}(g_V(x))\cdot \prod_{\alpha \in V\backslash\{0\}}\alpha.$ Hence $\deg(z)=\deg\left(\frac{df}{dx}\right)p^t$. If $\gcd(d,q)=1$, we can conclude that $\deg(z)=(d-1)p^t$. Choosing $Y=Z$ and $Y_1=0$, we see from Corollary \ref{cor:3.5} that if $\gcd(d,p)=1$, we obtain $q$-ary $[dp^t,dp^t-2m]$-LCD MDS codes for $2 \le m \le p^t/2.$
}
\end{construction}

Collecting the results from the explicit constructions given above, Theorem \ref{thm:main} follows immediately.

\section{Acknowledgements}
The first author gratefully acknowledges the support from The Danish Council for Independent Research (Grant No.~DFF--4002-00367).
The second author acknowledges the support from the Open Research Fund of National Mobile Communications Research Laboratory, Southeast University (No. 2017D07), and the National Natural Science Foundation of China under Grant 11501117.
The authors would also like to thank Prof. Chaoping Xing for pleasant discussions concerning this work.

\vspace{2ex}
\noindent
Peter Beelen\\
Technical University of Denmark,\\
Department of Applied Mathematics and Computer Science,\\
Matematiktorvet 303B,\\
2800 Kgs. Lyngby,\\
Denmark,\\
pabe@dtu.dk

\vspace{2ex}
\noindent
Lingfei Jin\\
Fudan University,\\
School  of Computer Science,\\
825 Zhangheng Road, Shanghai, \\
China,\\
lfjin@fudan.edu.cn

\end{document}